\documentclass[11pt]{article}
\usepackage{pdfpages}
\usepackage{multirow}
\usepackage{array}
\usepackage{graphicx}
\usepackage{subcaption}
\usepackage{latexsym}
\usepackage{pstricks}
\usepackage{pst-node}
\usepackage{pst-tree}
\usepackage{url}
\usepackage{times}
\usepackage{helvet}
\usepackage{courier}
\usepackage{algorithmic}
\usepackage[ruled,vlined,linesnumbered]{algorithm2e}
\usepackage{amsthm}
\usepackage{graphicx,amssymb,amsmath}
\usepackage{mathrsfs}
\usepackage{authblk}
\usepackage{enumitem}
\usepackage{comment}
\usepackage{framed}

\usepackage[hypertexnames=false, colorlinks=true, citecolor=blue, linkcolor=red, urlcolor=black]{hyperref}

\usepackage{geometry}
\makeatletter
\renewcommand*{\verbatim@font}{\sffamily}

\title{Simple Concurrent Connected Components Algorithms\thanks{This paper is a revised and expanded version of \cite{liu_tarjan}. Research at Princeton University partially supported by an innovation research grant from Princeton and a gift from Microsoft. Some work was done at AlgoPARC workshops in 2017 and 2019, partially supported by NSF grants CCF-1930579 and 1745331, respectively.}}

\author{S. Cliff Liu\thanks{Department of Computer Science, Princeton University; sixuel@cs.princeton.edu.}
}

\author{Robert E. Tarjan\thanks{Department of Computer Science, Princeton University, and Intertrust Technologies; ret@cs.princeton.edu.}
}
\affil{}

\begin{document}

\maketitle

\newcounter{dummy}
\newtheorem{lemma}[dummy]{Lemma}
\newtheorem{definition}[dummy]{Definition}
\newtheorem{remark}[dummy]{Remark}
\newtheorem{theorem}[dummy]{Theorem}
\newtheorem{corollary}[dummy]{Corollary}
\newtheorem{observation}[dummy]{Observation}
\newtheorem{claim}{Claim}

\begin{abstract}
   We study a class of simple algorithms for concurrently computing the connected components of an $n$-vertex, $m$-edge graph. Our algorithms are easy to implement in either the COMBINING CRCW PRAM or the MPC computing model. For two related algorithms in this class, we obtain $\Theta(\lg n)$ step and $\Theta(m \lg n)$ work bounds.\footnote{We denote by $\lg$ the base-two logarithm.}
   For two others, we obtain $O(\lg^2 n)$ step and $O(m \lg^2 n)$ work bounds, which are tight for one of them. All our algorithms are simpler than related algorithms in the literature. We also point out some gaps and errors in the analysis of previous algorithms. Our results show that even a basic problem like connected components still has secrets to reveal.
\end{abstract}

\section{Introduction}\label{intro}

The problem of finding the connected components of an undirected graph with $n$ vertices and $m$ edges is fundamental in algorithmic graph theory.
Any kind of graph search, such as depth-first or breadth-first, solves it in linear time sequentially, which is best possible.
The problem becomes more interesting in a concurrent model of computation.
In the heyday of the theoretical study of PRAM (parallel random-access machine) algorithms, many more-and-more efficient algorithms for the problem were discovered, culmina-ting with the $O(\lg n)$ step, $O(m)$ work randomized algorithms of Halperin and Zwick \cite{DBLP:journals/jcss/HalperinZ96, halperin2001optimal}, the second of which computes spanning trees of the components.
The goal of most of the work in the PRAM model was to obtain the best asymptotic bounds, not the simplest algorithms.

With the growth of the internet, the world-wide web, and cloud computing, finding connected components on huge graphs has become commercially important, and practitioners have put versions of the PRAM algorithms into use. Many of the PRAM algorithms are complicated, and even some of the simpler ones have been further simplified when implemented.
Experiments suggest that the resulting algorithms perform\\
well in practice, but some of the published claims about their theoretical performance are incorrect or unjustified.

Given this situation, our goal here is to develop and analyze the simplest possible efficient algorithms for the problem and to rigorously analyze their efficiency. In exchange for algorithmic simplicity, we are willing to allow analytic complexity.
Our algorithms are easy to implement in either the COMBINING CRCW PRAM model \cite{DBLP:books/daglib} in which write conflicts are resolved in favor of the smallest value, or in the MPC (massive parallel computing) model \cite{DBLP:journals/jacm/BeameKS17}.

The COMBINING CRCW PRAM model is stronger than the more standard ARBITRARY CRCW PRAM model, in which write conflicts are resolved arbitrarily (one of the writes succeeds, but the algorithm has no control over which one), but is weaker than the MPC model.

We consider a class of simple deterministic algorithms for the problem. We study in detail four algorithms in the class, all of which are simpler than corresponding existing algorithms. For two we prove a step bound of $\Theta(\lg n)$ and a work bound of $\Theta(m \lg n)$. For the other two, we prove a step bound of $O(\lg^2 n)$ and a work bound of $O(m \lg^2 n)$. These bounds are tight for one of the two. We also show that one of our $O(\lg^2 n)$-step algorithms takes $O(d)$ steps where $d$ is the largest diameter of a component, but the others do not.

Our paper contains five sections in addition to this introduction.
\S{\ref{af_section}} presents our algorithmic framework and a general lower bound in the MPC model.
\S{\ref{alg_section}} presents our algorithms.
\S{\ref{sec_analysis}} proves upper and lower step bounds.
\S{\ref{rw_section}} discusses related work.
\S{\ref{remark_section}} contains final remarks and open problems.

\section{Algorithmic Framework}\label{af_section}

Given an undirected graph with vertex set $[n] = \{1, 2, ..., n\}$ and $m$ edges, we wish to compute its connected components via a concurrent algorithm.  More precisely, for each component we want to label all its vertices with a unique vertex in the component, so that two vertices are in the same component if and only if they have the same label.
To state bounds simply, we assume $n > 2$ and $m > 0$.
We denote an edge by the unordered pair of its ends.

As our computing model we use a COMBINING CRCW (concurrent read, concurrent write) PRAM (parallel random-access machine) \cite{DBLP:books/daglib}. Such a machine consists of a large common memory and a number of processes, each with a small private memory. In one step, each process can do one unit of local computation, read one word of common memory, or write into one word of common memory.  The processes operate in lockstep. Concurrent reads and writes are allowed, with write conflicts resolved in favor of the smallest value written. We discuss weaker variants of the PRAM model in \S{\ref{rw_section}}. We measure the efficiency of an algorithm primarily by the number of concurrent steps and secondarily by the \emph{work}, defined to be the number of steps times the number of processes.

Our algorithms are also easy to implement on the MPC (massively parallel computing) model \cite{DBLP:journals/jacm/BeameKS17}. This is a model of distributed computation based on the BSP (bulk synchronous parallel) model \cite{DBLP:journals/cacm/Valiant90}. The MPC model is more powerful than our PRAM model but is a realistic model of cloud computing platforms. An MPC machine consists of a number of processes, each with a private memory. There is no common global memory; the processes communicate with each other by sending messages. Computation proceeds in globally synchronized steps. In one step, each process receives the messages sent to it in the previous step, does some amount of local computation, and then sends a message or messages to one or more other processes.

We specialize the MPC model to the connected components problem as follows. There is one process per edge and one per vertex. A process can only send a message to another process once it knows about that process. Initially a vertex knows only about itself, and an edge knows only about its two ends. Thus in the first concurrent step only edges can send messages, and only to their ends. A vertex or edge knows about another vertex or edge once it has received a message containing the vertex or edge. This model ignores contention resulting from many messages being sent to the same process, and it allows a process to send many messages in one step.

The MPC model is quite powerful, but even in this model there is a non-constant lower bound on the number of steps needed to compute connected components.

\begin{theorem}\label{lower_bound}
    Computing connected components in the MPC model takes $\Omega(\lg d)$ steps, where $d$ is the largest diameter of a component.
\end{theorem}
\begin{proof}
    Let $u$ be the vertex that eventually becomes the label of all vertices in a component of diameter $d$.
    Some vertex $v$ is at distance at least $d/2$ from $u$. 
    An induction on the number of steps shows that after $k$ steps a vertex has only received messages containing vertices within distance $2^{k-1}$. Since $v$ must receive a message containing $u$, the computation takes at least $\lg d$ steps.
\end{proof}

It is easy to solve the problem in $O(\lg d)$ steps if messages can be arbitrarily large:
send each edge end to the other end, and then repeatedly send from each vertex all the vertices it knows to all its incident vertices \cite{DBLP:conf/icde/RastogiMCS13}.
If there is a large component, however, this algorithm is not practical for at least two reasons: it requires huge memory at each vertex, and the number of messages sent in the last step can be quadratic in $n$.
Hence we restrict the local memory of a vertex or edge to hold only a small constant number of vertices and edges.
We also restrict messages to hold only a small constant number of vertices and edges, along with an indication of the message type, such as a label request or a label update.
Our goal is a simple algorithm with a step bound of $O(\lg n)$.
(We discuss the harder goal of achieving a step bound of $O(\lg d)$ in \S{\ref{rw_section}}.)

We consider algorithms that maintain a label for each vertex $u$, initially $u$ itself. The algorithm updates labels step-by-step until none changes, after which all vertices in a component have the same label, which is one of the vertices in the component. At any given time the current labels define a digraph (directed graph) of out-degree one whose arcs lead from vertices to their labels.
We call this the \emph{label digraph}.
If these arcs form no cycles other than loops (arcs of the form $(u, u)$), then this digraph is a forest of trees rooted at the self-labeled vertices: the parent of $u$ is its label unless this label is $u$, in which case $u$ is a root. We call this the \emph{label forest}.
Each tree in the forest is a \emph{label tree}.

All our algorithms maintain the label digraph as a forest; that is, they maintain acyclicity except for self-labels.  (We know of only two previous algorithms that do not maintain the label digraph as a forest: see \S{\ref{rw_section}}.)  Henceforth we call the label of a vertex $u$ its \emph{parent} and denote it by $u.p$, and we call a label tree just a \emph{tree}.
A non-root vertex is a \emph{child} of its parent.
A vertex is a \emph{leaf} if it is a child but it has no children of its own.
A tree is \emph{flat} if the root is the parent of every child in the tree, and is a \emph{singleton} if the root is the only vertex. (Some authors call a flat tree a \emph{star}.)
The \emph{depth} of a tree vertex $x$ is the number of arcs on the path from $x$ to the root of the tree: a root has depth zero, a child of a root has depth one.
The \emph{depth of a tree} is the maximum of the depths of its vertices.

When changing labels, a simple way to guarantee acyclicity is to replace a parent only by a smaller vertex.  We call this \emph{minimum labeling}. (An equivalent alternative, \emph{maximum labeling}, is to replace a parent only by a larger vertex).
A minimum labeling algorithm stops with each vertex labeled by the smallest vertex in its component. All our algorithms do minimum labeling. They also maintain the invariant that all vertices in a tree are in the same component, as do all algorithms known to us. That is, they never create a tree containing vertices from two or more components. At the end of the computation there is one flat tree per component, whose root is the minimum vertex in the component.

\section{Algorithms}\label{alg_section}

We consider algorithms that consist of initialization followed by a main loop that updates parents and repeats until no parent changes.
Initialization consists of setting the parent of each vertex equal to itself. 
The following pseudocode does initialization: \\

\emph{initialize}: \\
\indent\indent \textbf{for} each vertex $v$ \textbf{do} $v.p = v$ \\

Since initialization is the same for all our algorithms, we omit it in the descriptions below and focus on the main loop. Each iteration of the loop does a \emph{connect} step, which updates parents using current edges, one or more \emph{shortcut} steps, each of which updates parents using old parents, and possibly an \emph{alter} step, which alters edges.

Our algorithms maintain the following \emph{connectivity invariant}: $v$ and $v.p$ are in the same component, as are $v$ and $w$ if $\{v, w\}$ is an edge.
(The latter statement is trivial without edge alteration; edge alteration maintains it.)
If $\{v, w\}$ is an edge, replacing the parent of $v$ by any ancestor of $w$ preserves the invariant, as does replacing the parent of $w$ by any ancestor of $v$.
This gives us many ways of doing a connect step. We focus on two.
The first is \emph{direct-connect}, which for each edge $\{v, w\}$, uses the minimum of $v$ and $w$ as a candidate for the new parent of the other.
The other is \emph{parent-connect}, which uses the minimum of the old parents of $v$ and $w$ as a candidate for the new parent of the old parent of the other.
We express these methods in pseudocode below. We have written all our pseudocode so that it is correct and produces
unambiguous results even if the loops run sequentially rather than concurrently and the vertices
and edges are processed in arbitrary order. We say more about this issue below. \\

\emph{direct-connect}: \\
\indent\indent \textbf{for} each edge $\{v, w\}$ \textbf{do} \\
\indent\indent\indent \textbf{if} $v > w$ \textbf{then} \\
\indent\indent\indent\indent $v.p = \min\{v.p, w\}$ \\
\indent\indent\indent \textbf{else} $w.p = \min\{w.p, v\}$ \\

The pseudocode for \emph{parent-connect} begins by computing $v.o$, the old parent of $v$, for each vertex
$v$.
It then uses these old parents to compute the new parent $v.p$ of each vertex $v$.
The new parent of a vertex $x$ is the minimum $w.o < x.o$ such that there is an edge $\{v, w\}$ with $v.o = x$, if there is
such an edge; if not, the parent of $x$ does not change. \\

\emph{parent-connect}: \\
\indent\indent \textbf{for} each vertex $v$ \textbf{do} \\
\indent\indent\indent $v.o = v.p$ \\
\indent\indent \textbf{for} each edge $\{v, w\}$ \textbf{do} \\
\indent\indent\indent \textbf{if} $v.o > w.o$ \textbf{then} \\
\indent\indent\indent\indent $v.o.p = \min\{v.o.p, w.o\}$ \\
\indent\indent\indent \textbf{else} $w.o.p = \min\{w.o.p, v.o\}$ \\

If all reads and comparisons occur before all writes, and all writes occur concurrently, the following simpler pseudocode has the same semantics as that for \emph{parent-connect}: \\

\textbf{for} each edge $\{v, w\}$ \textbf{do} \\
\indent\indent \textbf{if} $v.p > w.p$ \textbf{then} \\
\indent\indent\indent $v.p.p = \min\{v.p.p, w.p\}$ \\
\indent\indent \textbf{else} $w.p.p = \min\{w.p.p, v.p\}$ \\

On the other hand, if this simpler loop is executed sequentially, then in general the results depend on the order in which the edges are processed. Suppose for example there are two edges $\{x, y\}$ and $\{v, w\}$ such that $x.p = v$ and $v.p = z$.
In \emph{parent-connect}, $w.p$ is a candidate to be the new
parent of $z$.
That is, after the connect, the new parent of $z$ will be no greater than the old parent of $w$.
But in the simpler loop, if $\{x, y\}$ is processed before $\{v, w\}$, the processing of $\{x, y\}$ might change the parent of $v$ to a vertex other than $z$, thereby making $w.p$ no longer a candidate for the new parent of $z$. Even though we are primarily interested in global concurrency, we want our
algorithms and bounds to be correct in the more realistic setting in which the edges are processed one group at a time, with the group size determined by the number of available processes.

On a COMBINING PRAM, we can use the simple loop for \emph{parent-connect} using the simpler loop, since there is global concurrency.
Each process for an edge $\{v, w\}$ reads $v.p$ and $w.p$. If $v.p > w.p$ it reads $v.p.p$, tests if $w.p < v.p.p$, and if so writes $w.p$ to $v.p.p$; if $v.p \le w.p$ it reads $w.p.p$, tests if $v.p < w.p.p$, and if so writes $v.p$ to $w.p.p$. All the reads occur before all the writes, and all the writes
occur concurrently, with a write of smallest value succeeding if there is a conflict.

In the MPC model, each vertex stores its parent. To execute \emph{parent-connect}, each process for an edge $\{v, w\}$ requests $v.p$ and $w.p$.
If $v.p > w.p$ it sends $w.p$ to $v.p$; otherwise, it sends $v.p$ to $w.p$.
Each vertex then updates its parent to be the minimum of its old value and the smallest of the received values.  All our other loops can be similarly implemented on a COMBINING CRCW PRAM or in the MPC model.

A connect step of either kind can move a subtree from one tree to another.
We can prevent this by restricting connection so that it only updates parents of roots.
The following pseudocode implements such restrictions of \emph{direct-connect} and \emph{parent-connect}, which we call \emph{direct-root-connect} and \emph{parent-root-connect}, respectively: \\

\emph{direct-root-connect}: \\
\indent\indent \textbf{for} each vertex $v$ \textbf{do} \\
\indent\indent\indent $v.o = v.p$ \\
\indent\indent \textbf{for} each edge $\{v, w\}$ \textbf{do} \\
\indent\indent\indent \textbf{if} $v > w$ and $v = v.o$ \textbf{then} \\
\indent\indent\indent\indent $v.p = \min\{v.p, w\}$ \\
\indent\indent\indent \textbf{else if} $w = w.o$ \textbf{then} \\
\indent\indent\indent\indent $w.p = \min\{w.p, v\}$ \\

\emph{parent-root-connect}: \\
\indent\indent \textbf{for} each vertex $v$ \textbf{do} \\
\indent\indent\indent $v.o = v.p$ \\
\indent\indent\indent \textbf{for} each edge $\{v, w\}$ \textbf{do} \\
\indent\indent\indent\indent \textbf{if} $v.o > w.o$ and $v.o = v.o.o$ \textbf{then} \\
\indent\indent\indent\indent\indent $v.o.p = \min\{v.o.p, w.o\}$ \\
\indent\indent\indent\indent \textbf{else if} $w.o = w.o.o$ \textbf{then} \\
\indent\indent\indent\indent\indent $w.o.p = \min\{w.o.p, v.o\}$ \\

In \emph{direct-root-connect} (as in \emph{parent-connect}) we need to save the old parents to get a correct sequential implementation, so that the root test is correct even if the parent has been changed by processing another edge during the same iteration of the loop over the edges.
If we truly have global concurrency, simpler pseudocode suffices, as for \emph{parent-connect}.

Shortcutting replaces the parent of each vertex by its grandparent. The following pseudocode implements shortcutting: \\

\emph{shortcut}: \\
\indent\indent \textbf{for} each vertex $v$ \textbf{do} \\
\indent\indent\indent $v.o = v.p$ \\
\indent\indent \textbf{for} each vertex $v$ \textbf{do} \\
\indent\indent\indent $v.p = v.o.o$  \\

In the case of \emph{shortcut}, the simpler loop
``for each vertex $v$ do $v.p =v.p.p$''
produces correct results and preserves our time bounds, even though sequential execution of the code produces different results depending on the order in which the vertices are processed.
All we need is that the new parent of a vertex is no greater than its old grandparent. Thus the simpler loop might well be a better choice in practice.

Edge alteration deletes each edge $\{v, w\}$ and replaces it by $\{v.p, w.p\}$ if $v.p \neq w.p$. The following pseudocode implements alteration: \\

\emph{alter}: \\
\indent\indent \textbf{for} each edge $\{v, w\}$ \textbf{do} \\
\indent\indent\indent \textbf{if} $v.p = w.p$ \textbf{then} \\
\indent\indent\indent\indent delete $\{v, w\}$ \\
\indent\indent\indent \textbf{else} replace $\{v, w\}$ by $\{v.p, w.p\}$ \\

We shall study in detail four algorithms, whose main loops are given below:

\begin{itemize}

    \item Algorithm \textsf{S}: \textbf{repeat} \{\emph{parent-connect}; \textbf{repeat} \emph{shortcut} \textbf{until} no $v.p$ changes\} \textbf{until} no $v.p$ changes

    \item Algorithm \textsf{A}: \textbf{repeat} \{\emph{direct-connect}; \emph{shortcut}; \emph{alter}\} \textbf{until} no $v.p$ changes

    \item Algorithm \textsf{R}: \textbf{repeat} \{\emph{parent-root-connect}; \emph{shortcut}\} \textbf{until} no $v.p$ changes

    \item Algorithm \textsf{RA}: \textbf{repeat} \{\emph{direct-root-connect}; \emph{shortcut}; \emph{alter}\} \textbf{until} no $v.p$ changes
\end{itemize}

Even though other algorithms fall within our framework, we focus on these four because they are simple and we can prove good bounds for them.
Before deriving bounds, we make some remarks about these and similar algorithms. Each of the four algorithms is distinct: for each pair of algorithms, there is a graph on which the two algorithms do different parent changes, as one can verify case-by-case.
Use of \emph{direct-connect} or \emph{direct-root-connect} requires edge alteration to obtain a correct algorithm.
Although different, algorithms \textsf{R} and \textsf{RA} behave similarity, and we use the same techniques to obtain bounds for both of them.
In algorithm \textsf{S}, the inner loop makes all trees flat, so the connect only changes the parents of roots, making \textsf{S} equivalent to the algorithm obtained by replacing \emph{parent-connect} by \emph{parent-root-connect}.
The algorithm formed from \textsf{S} by replacing \emph{parent-connect} by \emph{direct-connect} and adding alter to the end of the main loop is also equivalent to \textsf{S}.
Algorithms \textsf{S}, \textsf{R}, and \textsf{RA} are \emph{monotone} in that once two vertices are in the same tree, they remain in the same tree. This is not true for algorithm \textsf{A}. We discuss other related algorithms in \S{\ref{rw_section}}.

We conclude this section with a proof that our algorithms are correct. We begin by establishing some properties of edge alteration. We call an iteration of the main loop a \emph{round}.
We call a vertex \emph{bare} if is not an edge end and \emph{clad} if it is.

To prove correctness, we need the following key result.
\begin{lemma}\label{path_lemma}
    After $k$ rounds of algorithm \textsf{A} or \textsf{RA},
    each vertex that is clad or a root has a path to the smallest vertex in its component.
    If the algorithm is \textsf{A}, there is such a path with at most $\max\{0, d - k\}$ edges.
\end{lemma}
\begin{proof}
    The proof is by induction on $k$.
    The lemma is true for $k = 0$.
    Let $k > 0$, let $x$ be a clad vertex or a root at the end of round $k$,
    and let $w$ be the smallest vertex in the same component as $x$.
    If $x = w$, the lemma holds for $x$ and $k$.
    Suppose $x > w$.
    If $x$ is a root just before the alteration in round $k$, let $u = x$; otherwise, let $u$ be a vertex such that $u.p = x$ and $u$ is clad just before the alteration in round $k$.
    Such a $u$ must exist since $x$ is clad at the end of round $k$.
    Since $u \ge x > w$, $x \neq w$.
    By the induction hypothesis, there is a path from $u$ to $w$ at the beginning of round $k$, and if the algorithm is \textsf{A} the path contains at most $\max\{0, d - k + 1\}$ edges.
    Let $\{v, w\}$ be the last edge on this path. The alteration in round $k$ converts this path to a path from $x$ to $w$.
    Since $\{v, w\}$ exists at the beginning of round $k$, $v.p = w$ after the connect in round $k$. Thus the alter in round $k$ deletes $\{v, w\}$, so if the algorithm is \textsf{A}, the path from $x$ to $w$ contains at most $\max\{0, d - k\}$ edges.
\end{proof}

\begin{lemma}\label{lem_arc_end_move_up}
    In algorithms \textsf{A} and \textsf{RA}, \emph{(\romannumeral1)} an alter makes all leaves bare;
    \emph{(\romannumeral2)} once a vertex is a leaf it stays a leaf;
    \emph{(\romannumeral3)} once a leaf is bare it stays bare; and
    \emph{(\romannumeral4)} every non-root grandparent is clad.
\end{lemma}
\begin{proof}
    Part (\romannumeral1) is immediate from the definition of alter.
    A vertex becomes a leaf either in a connect or in a shortcut. A shortcut does not make a leaf into a non-leaf. Any leaf existing at the beginning of a connect is bare by (\romannumeral1) and hence cannot be made a non-leaf by the connect. Neither a connect nor a shortcut can make a bare vertex clad. Thus (\romannumeral2) and (\romannumeral3) hold.

    We prove (\romannumeral4) by induction on the number of steps. No vertex is a grandparent initially, so (\romannumeral4) holds initially. Let $x$ be a vertex that is not a non-root grandparent before a connect; that is, $x$ is either a root or has no grandchildren. If $x$ acquires a parent or child during the connect, then it is clad, so the lemma holds for $x$ after the connect. Since all children of $x$ (if any) are leaves, they cannot become non-leaves during the connect by (\romannumeral1), so $x$ cannot become a grandparent as a result of one of its children acquiring a child. Thus the connect preserves (\romannumeral4). A vertex that is a non-root grandparent after a shortcut is also a non-root grandparent before the shortcut, so a shortcut preserves (\romannumeral4). Suppose $x$ is a non-root grandparent before an alter. Then $x$ is a great-great grandparent before the shortcut preceding the alter, so it has a non-root grandchild $y$ that is clad. The shortcut makes $y$ a child of $x$. By Lemma~\ref{path_lemma} there is a path from $y$ to the smallest vertex in its component. The alter transforms this path into a path from $x$ to the smallest vertex.
    Since $x$ is not a root, it is not the smallest vertex, so the path from $x$ contains at least one edge, making $x$ clad.
\end{proof}

\begin{lemma}\label{lem_leaf}
    In all our algorithms, a leaf stays a leaf.
\end{lemma}
\begin{proof}
    For \textsf{A} or \textsf{RA}, this is part (\romannumeral3) of Lemma~\ref{lem_arc_end_move_up}.
    For the other algorithms, a connect or shortcut cannot make a leaf into a parent.
\end{proof}

\begin{theorem}\label{correctness1}
    All our algorithms are correct.
\end{theorem}
\begin{proof}
    For each algorithm, a proof by induction on the number of parent changes and alterations (if any) shows that $v$ and $v.p$ are in the same component for each vertex $v$, as are $v$ and $w$ for each edge $\{v, w\}$.
    Thus all vertices in any tree are in the same component.

    Parents only decrease, so the parent function always defines a forest.
    There are at most $n (n - 1) / 2$ parent changes.
    In the following paragraph, we prove that there is at least one parent change in any round except the last one, so each algorithm terminates in at most $n (n - 1) / 2 + 1$ rounds.

    Consider the state just before a round.  If there is at least one non-flat tree, then either the connect step or the first shortcut step will change a parent.
    Suppose all trees are flat but the vertices in some component are in two or more trees.
    Let $T$ be the tree containing the smallest vertex in the component and $T'$ another tree in the component.
    It is immediate for algorithms \textsf{S} and \textsf{R} and follows from Lemma~\ref{path_lemma} for \textsf{A} and \textsf{RA} that there is a path of edges from the root of $T'$ to the root of $T$.
    Thus there is an edge $\{v, w\}$ with $w$ but not $v$ in $T$.
    Since all trees are flat, $v.p$ and $w.p$ are roots.
    In algorithms \textsf{A} and \textsf{RA}, $v$ and $w$ are roots by part (\romannumeral1) of Lemma~\ref{lem_arc_end_move_up}.
    In algorithms \textsf{S} and \textsf{R}, the connect step will make $w.p$ the parent of $v.p$; in algorithms \textsf{A} and \textsf{RA}, the connect step will make $w$ the parent of $v$.
    We conclude that the algorithm can only stop when all trees are flat and there is one tree per component.
\end{proof}

\section{Efficiency}\label{sec_analysis}

On a graph that consists of a path of $n$ vertices, all our algorithms take $\Omega(\lg n)$ steps, since this graph has one component of diameter $n$, and the general lower bound of Theorem~\ref{lower_bound} applies.
We prove worst-case step bounds of $O(\min\{d, \lg n\} \lg n)$ for \textsf{S}, $O(\min\{d, \lg^2 n\})$ for \textsf{A}, and $O(\lg n)$ for \textsf{R} and \textsf{RA}.
We also show that algorithm \textsf{S} can take $\Omega(\lg^2 n)$ steps, and algorithms \textsf{R} and \textsf{RA} can take $\Omega(\lg n)$ steps even on graphs with constant diameter $d$.

\subsection{Analysis of \textsf{S}}

\begin{theorem}\label{thm_s_log}
    Algorithm \textsf{S} takes $O(\min\{d, \lg n\} \lg n)$ steps.
\end{theorem}
\begin{proof}
    Since the maximum depth of any tree is $n - 1$, and a shortcut reduces the depth of a depth-$k$ tree to $\lceil k/2 \rceil$, the inner loop stops in $O(\lg n)$ iterations.
    Let $u$ be the smallest vertex in some component.
    We prove by induction on $k$ that after $k$ rounds every vertex in the component at distance $k$ or less from vertex $u$ has parent $u$, which implies that the algorithm stops in at most $d + 1$ rounds.
    This is true initially.
    Let $v$ be a vertex at distance $k$ from $u$.
    If $v.p \neq u$ just before round $k$, there is an edge $\{v, w\}$ such that $w.p = u$ by the induction hypothesis. After the connect step in round $k$, $v.p = u$.

    To obtain an $O(\lg n)$ bound on the number of rounds, we prove that if there are two or more trees in a component, two rounds reduce the number of such trees by at least a factor of two.
    Call a root \emph{minimal} if edges incident to its tree connect it only with trees having higher roots.  A connect step makes each non-minimal root into a non-root.
    Let $x$ be a minimal root, and let $\{v, w\}$ be an edge with $v$ in the tree rooted at $x$ and $w$ in a tree rooted at $y > x$.
    If the connect step in the round does not make $y$ a child of $x$, then $y.p < x$ after the round, causing $x$ to become a non-root in the next round.

    Suppose there are $k > 1$ roots in a component at the beginning of a round.
    Among the minimal roots, suppose there are $j$ that get a new child as a result of the connect step in the round.
    At least $\max\{k-j, j\} \ge k/2$ roots are non-roots after two rounds.
\end{proof}

We show by example that \textsf{S} can take $\Omega(\lg^2 n)$ steps.

If there is a tree of depth $2^k$ just before the inner loop in a round, this loop will take at least $k$ steps.
If every round produces a \emph{new} tree of depth $2^k$, and the algorithm takes $k$ rounds, the total number of steps will be $\Omega(k^2)$. Our bad example is based on this observation.
It consists of $k$ components such that running algorithm \textsf{S} for $i$ rounds on the $i$-th component produces a tree that contains a path of length $2^k$.

At the beginning of a round of the algorithm, we define the \emph{implied graph} to be the graph whose vertices are the roots and whose edges are the pairs $\{v.p, w.p\}$ such that $\{v, w\}$ is a graph edge and $v.p \neq w.p$. Restricted to the roots, the subsequent behavior of algorithm \textsf{S} on the original graph is the same as its behavior on the implied graph.
We describe a way to produce a given implied graph in a given number of rounds.

To produce a given implied graph $G$ with vertex set $[n]$ in one round, we start with the \emph{generator} $g(G)$ of $G$, defined to be the graph with vertex set $[2n]$, edges $\{v, v+n\}$ and $\{v+n, v\}$ for each $v \in [n]$, and an edge $\{v+n, w+n\}$ for each edge $\{v, w\}$ in $G$.
A round of algorithm \textsf{S} on $g(G)$ does the following:
the connect step makes $v+n$ a child of $v$ for $v \in [n]$, and the shortcuts do nothing.
The resulting implied graph has vertex set $[n]$ and an edge $\{v, w\}$ for each edge $\{v,w\}$ in $G$; that is, it is $G$.

To produce a given implied graph $G$ in $i$ rounds, we start with $g^i(G)$. An induction on the number of rounds shows that after $i$ rounds of algorithm \textsf{S}, $G$ is the implied graph.

Let $k$ be a positive integer, and let $P$ be the path of vertices $1, 2, ..., 2^k + 1$ and edges $\{i, i+1\}$ for $i \in [2^k]$.
Our bad example is the disjoint union of $P, g(P), g^2(P), ..., g^{k-1}(P)$, with the vertices renumbered so that each component has distinct vertices and the order within each component is preserved.
Algorithm \textsf{S} takes $\Omega(k^2)$ steps on this graph.
The number of vertices is $n = (2^k + 1) (2^k - 1) = 2^{2k} - 1$. The number of edges is $2^{2k} - k - 1$, since the graph is a set of $k$ trees.
Thus the number of steps is $\Omega(\lg^2 n)$.

\subsection{Analysis of \textsf{A}}\label{subsec_d_bound}

In analyzing \textsf{A}, \textsf{R}, and \textsf{RA}, we assume that the graph is connected:
this is without loss of generality since \textsf{A}, \textsf{R}, and \textsf{RA} operate independently and concurrently on each component, and each round does $O(1)$ steps.

\begin{theorem}\label{e_diameter}
    Algorithm \textsf{A} takes $O(d)$ steps.
\end{theorem}
\begin{proof}
    By Lemma~\ref{path_lemma}, after $d$ rounds there are no edges and only one tree.
    By Lemma~\ref{lem_arc_end_move_up}, this tree has depth at most two.
    Thus the algorithm stops after at most $d + 2$ rounds.
\end{proof}

A simple example shows that Theorem~\ref{e_diameter} is false for \textsf{S}, \textsf{R}, and \textsf{RA}.
Consider the graph whose edges are $\{i, i+1\}$ for $i \in [n-1]$ and $\{i, n\}$ for $i \in [n-1]$.
After the first connection step, there is one tree: $1$ is the parent of $2$ and $n$, and $i$ is the parent of $i+1$ for $i \in [2, n-2]$.
Subsequent connection steps do nothing; the tree only becomes flat after $\Omega(\lg n)$ shortcuts.

To obtain an $O(\lg^2 n)$ step bound for algorithm \textsf{A}, we show that $O(\lg n)$ rounds reduce the number of non-leaf vertices by at least a factor of $2$, from which an overall $O(\lg^2 n)$ step bound follows.

It is convenient to shift our attention from rounds to passes.
A \emph{pass} is the interval from the
beginning of one shortcut to the beginning of the next.
Pass $1$ begins with the shortcut in round $1$ and ends with the connect in round $2$.
We need one additional definition.
A vertex is \emph{deep} if it is a non-root with at least one child and all its children are leaves.

\begin{lemma} \label{lem_deep_leaf}
    A vertex that is deep at the beginning of a pass is a leaf at the end of the pass.
\end{lemma}
\begin{proof}
    Let $x$ be a vertex that is deep at the beginning of a pass.
    The shortcut in the pass makes $x$ a leaf. Once a vertex is a leaf, it stays a leaf by Lemma~\ref{lem_arc_end_move_up}.
\end{proof}


\begin{lemma}\label{lem_deep}
    Suppose there are at least two roots at the beginning of a pass, and that $x$ is a root all of whose children are bare leaves. Then $x$ is not a root after the pass.
\end{lemma}
\begin{proof}
    By Lemma~\ref{path_lemma}, at the beginning of the pass there is an edge $\{v, w\}$ with $v$ but not $w$ in the tree with root $x$. Since all children of $x$ are bare leaves, $v = x$. The edge $\{x, w\}$ existed at the beginning of the connect just before the pass. Since this connect did not make $x$ a non-root, $w > x$, and since $w$ is not a child of $x$ after the connect, $w.p < w$ after it. The alter in the pass replaces $\{x, w\}$ by $\{x, w.p\}$. The connect in the pass then makes $x$ a non-root.
\end{proof}

We need one more idea, which we borrow from the analysis of \emph{path halving}, a method used in disjoint set union algorithms that shortcuts a single path \cite{DBLP:journals/jacm/TarjanL84}.
For any vertex $v$, we define the \emph{level} of $v$ to be $v.l = \lfloor \lg(v - v.p) \rfloor$ unless $v.p = v$, in which case $v.l = 0$.
The level of a vertex is at least $0$, less than $\lg n$, and non-decreasing.
The following lemma quantifies the effect of a shortcut on a sufficiently long path.

\begin{lemma}\label{level_sum_increase}
    Assume $n \ge 4$. Consider a tree path $P$ of $k \ge 4 \lg n$ vertices.
    A shortcut increases the sum of the levels of the vertices on $P$ by at least $k/4$.
\end{lemma}
\begin{proof}
    Let $u$, $v$, and $w$ be three consecutive vertices on $P$, with $v$ the parent of $u$ and $w$ the parent of $v$.
    Let $i$ and $j$ be the levels of $u$ and $v$, respectively.
    A shortcut increases the level of $u$ from $i$ to $\lfloor \lg(u - w) \rfloor = \lfloor \lg(u - v + v - w) \rfloor \ge \lfloor \lg(2^i + 2^j) \rfloor$.
    If $i < j$, this increases the level of $u$ by at least $j - i$; if $i = j$, it increases the level of $u$ by one.

    Let $x_1, x_2, ..., x_{k-1}$ be the vertices on $P$ from largest to smallest (deepest to shallowest), excluding the last one.
    For each $i \in [k - 2]$, let $\Delta_i = x_{i+1}.l - x_i.l$.
    The sum of the $\Delta_i$'s is $\Sigma = x_{k-1}.l - x_1.l \ge - \lg n$ since the sum telescopes.
    Let $k_{+}$, $k_0$, and $k_{-}$, respectively, be the number of positive, zero, and negative $\Delta_i$'s, and let $\Sigma_{+}$ and $\Sigma_{-}$ be the sum of the positive $\Delta_i$'s and the sum of the negative $\Delta_i$'s, respectively.
    By the previous paragraph, the sum of the levels of the vertices on $P$ increases by at least $\Sigma_{+} + k_0$.

    From $\Sigma = \Sigma_{+} + \Sigma_{-}$ we obtain $\Sigma_{+} > - \Sigma_{-} - \lg n$.
    Since the $\Delta_i$'s are integers, $\Sigma_{+} \ge k_{+}$ and $- \Sigma_{-} \ge k_{-}$.
    Thus $2 (\Sigma_{+} + k_0) \ge 2 \Sigma_{+} + k_0 \ge k_{+} + k_0 + k_{-} - \lg n = k - 2 - \lg n \ge k / 2$, since $n \ge 4$ implies $k / 2 \ge 2 \lg n \ge \lg n + 2$.
    Dividing by two gives $\Sigma_{+} + k_0 \ge k / 4$.
\end{proof}

We combine Lemmas \ref{lem_deep_leaf}, \ref{lem_deep}, and \ref{level_sum_increase} to obtain the desired result.


\begin{lemma}\label{green_reduce_constant}
    Assume $n \ge 4$. Suppose that at the beginning of pass $i$ there are at least two roots and more than $k/2$ but at most $k$ non-leaves.
    After $O(\lg n)$ passes there are at most $k / 2$ non-leaves or at most one root.
\end{lemma}
\begin{proof}
    Assume the hypotheses of the lemma are true. Call a vertex \emph{fresh} if it is a non-leaf at the beginning of pass $i$ and \emph{stale} otherwise. After pass $i$, all the stale leaves are bare by Lemma~\ref{lem_arc_end_move_up}.
    Suppose the hypotheses of the lemma hold at the beginning of some pass after pass $i$. At least one of the following four cases occurs:
    \begin{enumerate}[label=(\roman*)]
        \item 
            There are at least $k/8$ clad leaves. One pass makes all clad leaves bare by Lemma~\ref{lem_arc_end_move_up}.
            Since each clad leaf is fresh and can only become bare once, there are at most $k / (k/8) = 8$ passes in which this case can occur.
            \label{green_case_1}
        \item There are at least $k/8$ roots of flat trees, all of whose children are bare. One pass makes all such roots non-roots by Lemma~\ref{lem_deep}. There are at most $k/(k/8) = 8$ passes in which this case can occur. \label{green_case_2}
        \item There are at least $k / (32 \lg n)$ deep vertices. This pass makes all these vertices into leaves by Lemma~\ref{lem_deep_leaf}. There are at most $k / (k / (32 \lg n)) = 32 \lg n$ passes in which this case can occur. \label{green_case_3}
        \item None of the first three cases occurs.
            Since Cases \ref{green_case_1} and \ref{green_case_2} do not occur, there are at most $k/4$ roots of flat trees.
            Thus there are at least $k/4$ non-leaves in non-flat trees.
            Since Case~\ref{green_case_3} does not occur, there are at most
            $k / (32 \lg n)$ deep vertices.
            From each of these deep vertices there is a tree path to a root.
            Every non-leaf in a non-flat tree is on one or more such paths. Find a deep vertex whose tree path is longest, and delete this path. This may break the tree containing the path into several trees, but this does not matter: it does not increase the number of deep vertices, although it may convert some deep vertices into roots, which only improves the bound.
            Repeat this process until at least $k/8$ non-leaves are on deleted paths.
            Each path contains at least $(k/8) / (k / (32 \lg n)) = 4 \lg n$ vertices.
            By Lemma~\ref{level_sum_increase}, the shortcut increases the sum of the levels of the vertices on these paths by at least $k / 32$.
            There are at most $(k \lg n) / (k/32) = 32 \lg n$ passes in which this case can occur.
    \end{enumerate}
    We conclude that after at most $16 + 64 \lg n$ passes, either there are at most $k/2$ non-leaves or at most one root.
\end{proof}

\begin{theorem}\label{main_result2}
    Algorithm \textsf{A} takes $O(\lg^2 n)$ steps.
\end{theorem}
\begin{proof}
    The theorem is immediate from Lemma~\ref{green_reduce_constant}, since once there is a single root the algorithm stops after $O(\lg n)$ steps.
\end{proof}

We do not know whether the bound in Theorem~\ref{main_result2} is tight.
We conjecture that it is not, and that algorithm \textsf{A} takes $O(\lg n)$ steps.
We \emph{are} able to prove an $O(\lg n)$ bound for the monotone algorithms \textsf{R} and \textsf{RA}, which we do in the next section.

An algorithm similar to algorithm \textsf{A} is algorithm \textsf{P}, which replaces \emph{direct-connect} in \textsf{A} by \emph{parent-connect} and deletes \emph{alter}.
The following pseudocode implements the main loop of this algorithm: \\

Algorithm \textsf{P}: \textbf{repeat} \{\emph{parent-connect}; \emph{shortcut}\} \textbf{until} no $v.p$ changes \\

We conjecture that this algorithm, too, has an $O(\lg n)$ step bound but are unable to prove even an $O(\lg^2 n)$ bound, the problem being that Lemma~\ref{lem_deep} is false for this algorithm.

\subsection{Analysis of \textsf{R} and \textsf{RA}} \label{subsection_lsb}

Lemma~\ref{lem_deep} does not hold for algorithms \textsf{R} and \textsf{RA}: a flat tree can remain unchanged for a non-constant number of rounds.
But we can obtain an even better bound than that of Theorem~\ref{main_result2} by using a different analytical technique, that of Awerbuch and Shiloach \cite{DBLP:journals/tc/AwerbuchS87}, extended to cover a constant number of rounds rather than just one.

We call a tree \emph{passive} in a round if it exists both at the beginning and at the end of the round; that is, the round does not change it.
A passive tree is flat, but a flat tree need not to be passive.
We call a tree \emph{active} in a round if it exists at the end of the round but not at the beginning.
An active tree contains at least two vertices and has depth at least one.

We say a connect \emph{links} trees $T$ and $T'$ if it makes the root of one of them a child of a vertex in the other.
If the connect makes the root of $T$ a child of a vertex in $T'$, we say the connect \emph{links $T$ to $T'$}.

\begin{lemma}\label{lem_link}
    If trees $T$ and $T'$ are passive in round $k$, then there is no edge with one end in $T$ and the other end in $T'$, and the connect in round $k + 1$ does not link $T$ and $T'$.
\end{lemma}
\begin{proof}
    If $T$ and $T'$ were linked in round $k + 1$, there would be an edge connecting them that caused the link.
    In algorithm \textsf{RA}, Lemma~\ref{lem_arc_end_move_up} implies that any such edge connects the roots of $T$ and $T'$ at the beginning of round $k$.
    Thus in either algorithm $T$ and $T'$ would have been linked in round $k$, contradicting their passivity in round $k$.
\end{proof}

If $T$ exists at the end of round $k$, its \emph{constituent trees} at the end of round $j \le k$ are the trees existing at the end of round $j$ whose vertices are in $T$.
Since algorithms \textsf{R} and \textsf{RA} are monotone, these trees partition the vertices of $T$.

\begin{lemma}\label{lem_constituent_active}
    Let $T$ be an active tree in round $k$. Then for $j \le k$ at least one of the constituent trees of $T$ in round $j$ is active.
\end{lemma}
\begin{proof}
    The proof is by induction on $j$ for $j$ decreasing.
    The lemma holds for $j=k$ by assumption.
    Suppose it holds for $j > 0$. If the constituent trees of $T$ in round $j-1$ were all passive, the connect step in round $j$ would change none of them by Lemma~\ref{lem_link}.
    Neither would the shortcut in round $j$, contradicting the existence of an active constituent tree in round $j$.
\end{proof}

Since algorithm \textsf{RA} is slightly simpler to analyze than \textsf{R}, we first analyze \textsf{RA}, and then discuss the changes needed to make the analysis apply to \textsf{R}.
We measure progress using the potential function of Awerbuch and Shiloach, modified so that it is non-increasing and passive trees have zero potential.
In \textsf{RA} we define the \emph{individual potential} of a tree $T$ at the end of round $k$ to be zero if $T$ is passive in round $k$, or two plus the maximum of zero and the maximum depth of an arc end in $T$ if $T$ is active in round $k$.
If $T$ exists at the end of round $k$ and $j \le k$, we define the \emph{potential} $\Phi_j(T)$ of $T$ at the end of round $j$ to be the sum of the individual potentials of its constituent trees at the end of round $j$.
We define the \emph{total potential} at the end of round $k$ to be the sum of the potentials of the trees existing at the end of round $k$.

We shall prove that the total potential decreases by a constant factor in a constant number of rounds.
This will give us an $O(\lg n)$ bound on the number of rounds.
It suffices to consider each active tree individually.

\begin{lemma}\label{lem_high_tree}
    Let $T$ be active in round $k > 1$ of \textsf{RA}.
    Then $\Phi_{k-1}(T) \ge \Phi_k(T)$.
    If $\Phi_{k-1}(T) \ge 4$ then $\Phi_{k-1}(T) \ge (5/4) \Phi_k(T)$.
\end{lemma}
\begin{proof}
    Let $t \ge 1$ be the number of active constituent trees of $T$ in round $k-1$, and let $\ell = \Phi_{k-1}(T) - 2t$. Then $\ell \ge 0$.
    Consider the tree $S$ formed from the constituent trees of $T$ in round $k-1$ by the connect step in round $k$.
    The shortcut in round $k$ transforms $S$ into $T$.
    By Lemma~\ref{lem_link}, along any path in $S$ there cannot be consecutive vertices from two different passive trees.
    By Lemma~\ref{lem_arc_end_move_up}, at the beginning of round $k$ no edge end is a leaf, so the deepest edge end in $S$ has depth at most $\ell+2t$: its path to the root contains at most $\ell+t$ vertices in active constituent trees and at most $t+1$ vertices (all roots) in passive constituent trees.
    The shortcut of $S$ in round $k$ reduces the maximum depth of an arc end to at most $\lceil \ell/2 \rceil + t$. The alter in round $k$ reduces this maximum depth by at least one, to at most $\lceil \ell/2 \rceil + t - 1$.
    Thus $\Phi_k(T) \le \lceil \ell/2 \rceil + t - 1$.
    Since $\Phi_{k-1}(T) = \ell + 2t$ and $t \ge 1$, $\Phi_{k-1}(T) \ge \Phi_k(T)$, giving the first part of the lemma.

    We prove the second part of the lemma by induction on $\Phi_{k-1}(T) = \ell + 2t$.
    If $\Phi_{k-1}(T) = 4$, then $t = 2$ and $\ell=0$, or $t=1$ and $\ell=2$, so $\Phi_k(T) \le \lceil \ell/2 \rceil + t + 1 = 3$.
    If $\Phi_{k-1}(T) = 5$, then $t = 2$ and $\ell=1$, or $t=1$ and $\ell=3$, so $\Phi_k(T) \le 4$. In both cases the second part of the lemma is true.
    Each increase of $\ell$ by two or $t$ by one increases $\Phi_{k-1}(T)$ by two and increases the upper bound of $\lceil \ell / 2 \rceil + t + 1$ on $\Phi_k(T)$ by one, which preserves the inequality $\Phi_{k-1}(T) \ge (5/4) \Phi_k(T)$.
\end{proof}

Lemma~\ref{lem_high_tree} gives a potential drop for any active tree $T$ such that $\Phi_{k-1}(T) \ge 4$.
To obtain a potential drop if $\Phi_{k-1}(T) < 4$, we need to consider two rounds if $\Phi_{k-1}(T) = 3$ and three rounds if $\Phi_{k-1}(T) = 2$.

\begin{lemma}\label{lem_potential_3}
    Let $T$ be an active tree in round $k > 2$ of \textsf{RA} such that $\Phi_{k-1}(T) = 3$. Then $\Phi_{k-2}(T) \ge (4/3) \Phi_k(T)$.
\end{lemma}
\begin{proof}
    The lemma holds if $T$ has at least two active constituent trees in round $k-2$ or one with an edge end of depth two or more, since then $\Phi_{k-2}(T) \ge 4$.
    Suppose neither of these cases occurs. Then $T$ has one active constituent tree, say $T_2$, in round $k-2$.
    By Lemma~\ref{lem_link}, no pair of passive constituent trees of $T$ in round $k-2$ are connected by an edge. Thus they are all connected by an edge with the root or a child of the root of $T_2$.
    Let $T_1$ be the active constituent tree of $T$ in round $k-1$. The root of $T_2$ is the root or a child of the root of $T_1$. It follows that each passive constituent tree of $T$ in round $k-1$ has an edge connecting its root with that of $T_1$.
    Hence the tree containing the vertices of $T_1$ formed by the connect step in round $k$ has no edge ends of depth greater than two, which implies that $T$ has no edge ends of positive depth, making its potential two.
\end{proof}

\begin{lemma}\label{lem_potential_2}
    Let $T$ be an active tree in round $k > 3$ of \textsf{RA} such that $\Phi_{k-1}(T) = 2$. Then $\Phi_{k-3}(T) \ge (3/2) \Phi_k(T)$.
\end{lemma}
\begin{proof}
    The lemma holds if there are at least two active constituent trees of $T$ in round $k-2$ or in round $k-3$, or there is a constituent tree in one of these rounds with an edge end of positive depth: two active trees have total potential at least four, and an active tree with an edge end of depth at least one has a potential of at least three.

    Suppose not. Let $T_3$, $T_2$, and $T_1$ be the active constituent trees of $T$ in rounds $k-3$, $k-2$, and $k-1$, respectively. Since the constituent trees of $T$ in round $k-3$ other than $T_3$ are passive, no edge connects any pair of them by Lemma~\ref{lem_link}. Since all their vertices are in $T$, each such tree must have an edge connecting its root with that of $T_3$. The connect step in round $k-2$ makes the minimum vertex in $T$ the root of $T_2$. Tree $T_2$ has depth at most two by Lemma~\ref{lem_arc_end_move_up}, since only its root can be an edge end.
    The connect step in round $k-1$ links each passive constituent tree of $T$ in round $k-2$ to $T_2$, forming a tree $S_1$ containing all the vertices of $T$ and of depth at most two. The shortcut in round $k-1$ transforms $S_1$ to $T_1$, which is flat since $S_1$ has depth at most two.
    But since $T_1$ is flat and contains all the vertices in $T$, it must be passive in round $k$, a contradiction.
\end{proof}

Having covered all the cases, we are ready to put them together.
Let $a = (3/2)^{1/3} = 1.1447+$, and let $|T|$ be the number of vertices in tree $T$.

\begin{lemma}\label{lem_tree_potential}
    Let $T$ be an active tree in round $k$ of \textsf{RA}.
    Then $\Phi_k(T) \le (3/2) |T| / a^{k-3}$.
\end{lemma}
\begin{proof}
    The proof is by induction on $k$. Since $T$ contains at least two vertices and has potential at most $|T| + 1$, it has potential at most $(3/2) |T|$.
    This gives the lemma for $k \le 3$.
    Let $k > 3$ and suppose the lemma holds for smaller values. We consider three cases.
    If $T$ contains an edge end of depth at least two, then $\Phi_k(T) \le (4/5) \Phi_{k-1}(T) \le (4/5) (3/2) |T| / a^{k-4} \le (3/2) |T| / a^{k-3}$ by Lemma~\ref{lem_high_tree}, the induction hypothesis, the linearity of the total potential, and the inequality $a \le 5/4$.
    If the maximum depth of an edge end in $T$ is one, then $\Phi_k(T) \le (3/4) \Phi_{k-2}(T) \le (3/4) (3/2) |T| / a^{k-5} \le (3/2) |T| / a^{k-3}$ by Lemma~\ref{lem_potential_3}, the induction hypothesis, the linearity of the total potential, and the inequality $a \le (4/3)^{1/2}$.
    If no vertex in $T$ other than the root is an edge end, then $\Phi_k(T) \le (2/3) \Phi_{k-3}(T) \le (2/3) (3/2) |T| / a^{k-6} \le (3/2) |T| / a^{k-3}$ by Lemma~\ref{lem_potential_2}, the induction hypothesis, the linearity of the total potential, and $a = (3/2)^{1/3}$.
\end{proof}

\begin{theorem}\label{thm_ra}
    Algorithm \textsf{RA} takes $O(\lg n)$ steps.
\end{theorem}
\begin{proof}
    By Lemma~\ref{lem_tree_potential}, if $k$ is such that $a^{k-3} > (3/2) n$, then no tree can be active in round $k$. Only the last round has no active trees.
\end{proof}

We can use the same approach to prove an $O(\lg n)$ step bound for algorithm \textsf{R}, but the details are more complicated.
Algorithm \textsf{RA} has the advantage over \textsf{R} that the alter step in effect does extra flattening.
If $T$ is active in round $k$ and $T_1$ is the only active constituent tree of $T$ in round $k-1$, it is possible for the depth of $T_1$ to be one and that of $T$ to be two.
This is not a problem in \textsf{RA}, because the alter decreases the depth of each edge end by one.
But in \textsf{R} we need to give extra potential to trees of depth one to make the potential function non-increasing.

In \textsf{R} we define the individual potential of a tree $T$ at the end of round $k$ to be zero if $T$ is passive in round $k$, or the depth of $T$ if $T$ is active in round $k$ and not flat, or two if $T$ is active in round $k$ and flat.
As in \textsf{RA}, if $T$ exists at the end of round $k$ and $j \le k$, we define the potential $\Phi_j(T)$ of $T$ at the end of round $j$ to be the sum of the potentials of its constituent trees at the end of round $j$, and we define the total potential at the end of round $k$ to be the sum of the potentials of the trees existing at the end of round $k$.
We prove analogues of Lemmas~\ref{lem_high_tree}-\ref{lem_tree_potential} and Theorem~\ref{thm_ra} for \textsf{R}.

\begin{lemma}\label{lem_r_high_tree}
    Let $T$ be active in round $k > 1$ of \textsf{R}.
    Then $\Phi_{k-1}(T) \ge \Phi_k(T)$, and if $\Phi_{k-1}(T) \ge 4$, then $\Phi_{k-1}(T) \ge (5/4) \Phi_k(T)$.
\end{lemma}
\begin{proof}
    Let $\ell$ be the sum of the depths of the active constituent trees of $T$ in round $k-1$, let $t$ be the number of these trees, and let $f$ be the number that are flat.
    Then $\ell \ge t$ and $\Phi_{k-1}(T) = \ell + f \ge \ell$.
    Let $S$ be the tree formed from the constituent trees of $T$ by the connect step in round $k$. This step in round $k$ does not make a leaf into a non-leaf, nor does it make a root of a passive tree the parent of another such root by Lemma~\ref{lem_link}.
    It follows that any path in $S$ contains at most $\ell+2$ vertices: at most $\ell-t$ non-leaf vertices of active constituent trees, at most $t+1$ roots of passive constituent trees, and at most one leaf of some constituent tree.
    The shortcut in round $k$ transforms $S$ into $T$, so the depth of $T$ is at most $\lceil \ell/2 \rceil + 1$, as is its potential: if $\ell=1$, $\lceil \ell/2 \rceil + 1 = 2$, which is the potential of a flat active tree.

    If $\ell = 1$, there is one active constituent tree, and it is flat, so $\Phi_{k-1}(T) = 2 = \Phi_k(T)$, making the lemma true.
    If $\ell \ge 2$, $\Phi_{k-1}(T) \ge \ell \ge \lceil \ell/2 \rceil + 1 \ge \Phi_{k}(T)$.
    If $\ell = 4$, $\Phi_{k-1}(T) \ge 4$ and $\Phi_k(T) \le 3$.
    If $\ell=5$, $\Phi_{k-1}(T) \ge 5$ and $\Phi_k(T) \le 4$.
    Thus the lemma is true if $\ell \le 5$.
    Each increase of $\ell$ by two increases the lower bound of $\ell$ on $\Phi_{k-1}(T)$ by two and increases the upper bound of $\lceil \ell/2 \rceil + 1$ on $\Phi_k(T)$ by one, which preserves the inequality $\Phi_{k-1}(T) \ge (5/4) \Phi_k(T)$, so the lemma holds for all $\ell$ by induction.
\end{proof}

Lemma~\ref{lem_r_high_tree} gives a potential drop for any active tree $T$ such that $\Phi_{k-1}(T) \ge 4$.
To obtain a potential drop if $\Phi_{k-1}(T) < 4$, we need to consider two rounds if $\Phi_{k-1}(T) = 3$ and five rounds if $\Phi_{k-1}(T) = 2$.

\begin{lemma}\label{lem_r_potential_3}
    Let $T$ be an active tree in round $k > 2$ of \textsf{R} such that $\Phi_{k-1}(T) = 3$.
    Then $\Phi_{k-2}(T) \ge (4/3) \Phi_k(T)$.
\end{lemma}
\begin{proof}
    If the constituent trees of $T$ in round $k-1$ or in round $k-2$ include at least two active trees, or one active tree of depth at least four,
    or $T$ has depth at most two, the lemma holds.

    Suppose not.
    Let $T_2$ and $T_1$ be the unique active constituent trees of $T$ in rounds $k-2$ and $k-1$, respectively.
    Let $S_1$ and $S$ be the trees containing the vertices of $T_2$ and $T_1$ formed by the connect steps in rounds $k-2$ and $k-1$, respectively.
    Trees $T_2$, $T_1$, and $T$ all have depth three, and $S_1$ and $S$ have depth five.
    No pair of passive constituent trees of $T$ in round $k-2$ are connected by an edge, so each such tree is connected to $T_2$ by an edge.
    Call such a tree \emph{primary} if it has an edge connecting it with the root or a child of the root of $T_2$, and \emph{secondary} otherwise.

    Since $S_1$ has depth five and $T_2$ has depth three, their roots must be different, so the root of $S_1$ is the minimum of the roots of the primary trees. Each vertex in $T_2$ that is the end of an edge whose other end is in a secondary tree has depth at least two in $T_2$, depth at least three in $S_1$, and depth at least two in $T_1$, which is formed from $S_1$ by the shortcut in round $k-1$. Since $T_1$ has depth three and $S$ has depth five, their roots must be different. But then the root of $S$ must be the root of one of the secondary trees, which is impossible since the root of $T_1$ is not the parent of any of the edge ends connecting $T_1$ with the secondary trees.
\end{proof}

\begin{lemma}\label{lem_r_potential_2}
    Let $T$ be an active tree in round $k > 5$ of \textsf{R} such that $\Phi_{k-1}(T) = 2$. Then $\Phi_{k-5}(T) \ge (3/2) \Phi_k(T)$.
\end{lemma}
\begin{proof}
    If for some $j$ between $k-5$ and $k-1$ inclusive the constituent trees of $T$ in round $j$ include at least two active trees, or one active tree of depth at least three, then $\Phi_j(T) \ge 3$, so the lemma holds.

    Suppose not. Then the constituent trees of $T$ in each round from $k-5$ to $k$ include exactly one active tree, of depth one or two.
    For $j$ between $1$ and $5$ inclusive let $T_j$ be the active constituent tree of $T$ in round $k-j$, for $j$ between $1$ and $4$ inclusive let $S_j$ be the tree containing the vertices of $T_{j+1}$ formed by the connect step in round $k-j$,
    and let $S$ be the tree containing the vertices of $T_1$ formed by the connect in round $k$.
    For $j$ from $1$ to $4$ inclusive, the shortcut in round $k-j$ transforms $S_j$ into $T_j$, and the shortcut in round $k$ transforms $S$ into $T$.

    No edge connects two passive constituent tree of $T$ in round $k-5$, so each such tree has an edge connecting it with $T_5$.
    Call such a tree $T$ \emph{primary} if it has an edge connecting it with the root of $T_t$ or to a child of the root of $T_5$, \emph{secondary} otherwise.
    Since $T_5$ has depth at most two, each secondary tree has an edge connecting it with a grandchild of the root of $T_5$.

    We consider two cases: the roots of $T_5$ and $T_4$ are the same, or they are different.
    In the former case, the roots of all primary trees are greater than the root of $T_5$, and the connect in round $k-4$ makes all of them children of the root of $T_5$.
    In the latter case, the root of $T_4$ is the minimum of the roots of the primary trees, and each such tree other than the one of minimum root is linked to $T_5$ in round $k-4$ or to $T_4$ in round $k-3$.

    Now consider the secondary trees. If the roots of $T_5$ and $T_4$ are the same, then after the shortcut in round $k-4$ each secondary tree has an edge connecting it with the root or a child of the root of $T_4$. By the argument in the preceding paragraph, each such tree will be linked with $T_4$ in round $k-3$ or with $T_3$ in round $k-2$.
    If the roots of $T_5$ and $T_4$ are different, none of the secondary trees has an edge connecting it with the root or a child of the root of $T_4$ at the end of round $k-4$. In this case the roots of $T_4$ and $T_3$ must be the same, so after the shortcut in round $k-3$ each secondary tree has an edge connecting it with the root or a child of the root of $T_3$. Each such tree will be linked with $T_3$ in round $k-2$ or with $T_2$ in round $k-1$. Furthermore the roots of $T_2$ and $T_1$ must be the same.

    It follows that there is only one constituent tree of $T$ in round $k-1$, and this tree is flat.
    But this tree must be $T$, making $T$ passive in round $k$, a contradiction.
\end{proof}

Let $b = (3/2)^{1/5} = 1.0844+$.

\begin{lemma}\label{lem_r_tree_potential}
    Let $T$ be an active tree in round $k$. Then $\Phi_k(T) \le (3/2) |T| / b^{k-5}$.
\end{lemma}
\begin{proof}
    The proof is by induction on $k$. Since $T$ contains at least two vertices and has potential at most $|T| + 1$, it has potential at most $(3/2) |T|$.
    This gives the lemma for $k \le 5$.
    Suppose $k > 5$ and suppose the lemma holds for smaller values. We consider three cases. If the depth of $T$ exceeds three, then $\Phi_k(T) \le (4/5) \Phi_{k-1}(T) \le (4/5) (3/2) |T| / b^{k-6} \le (3/2) |T| / b^{k-5}$ by Lemma~\ref{lem_r_high_tree}, the induction hypothesis, the linearity of the total potential, and the inequality $b \le 5/4$. If the depth of $T$ is three, then $\Phi_k(T) \le (3/4) \Phi_{k-2}(T) \le (3/4) (3/2) |T| / b^{k-7} \le (3/2) |T| / b^{k-5}$ by Lemma~\ref{lem_r_potential_3}, the induction hypothesis, the linearity of the total potential, and the inequality $b \le (4/3)^2$.
    If the depth of $T$ is at most two, then $\Phi_k(T) \le (2/3) \Phi_{k-5}(T) \le (2/3) (3/2) |T| / b^{k-10} \le (3/2) |T| / b^{k-5}$ by Lemma~\ref{lem_r_potential_2}, the induction hypothesis, the linearity of the total potential, and $b \le (3/2)^{1/5}$.
\end{proof}

\begin{theorem}\label{thm_r}
    Algorithm \textsf{R} takes $O(\lg n)$ steps.
\end{theorem}
\begin{proof}
    By Lemma~\ref{lem_r_tree_potential}, if $k$ is such that $b^{k-5} > (3/2) n$, then no tree can be active in round $k$. Only the last round has no active trees.
\end{proof}

For the variants of \textsf{R} and \textsf{RA} that do two shortcuts in each round instead of just one, we can simplify the analysis and improve the constants.
For \textsf{RA}, we let the potential of an active tree be the maximum depth of an edge end (or zero if there are no edge ends) plus one.
The potential of an active tree drops from the previous round by at least a factor of two unless it has only one active constituent tree in the previous round and that tree has potential one.
The proof of Lemma~\ref{lem_potential_2} gives a potential reduction of at least a factor of two in at most three rounds in this case.
Lemma~\ref{lem_tree_potential} holds with $a$ replaced by $2^{1/3} = 1.2599+$.
For \textsf{R}, we let the potential of an active tree be its depth.
The potential of an active tree drops from the previous round by at least a factor of at least two unless it has only one active constituent tree in the previous round and that tree is flat.
The proof of Lemma~\ref{lem_r_potential_2} gives a potential reduction of at least a factor of two in at most three rounds, and Lemma~\ref{lem_r_tree_potential} holds with $b$ replaced by $2^{1/3}$, the same constant as for the two-shortcut variant of \textsf{RA}.

This analysis suggests that doing two shortcuts per round rather than one might improve the practical performance of \textsf{R} and \textsf{RA}, especially since a shortcut needs only $n$ processes, but a connect step needs $m$.
Exactly how many shortcuts to do per round is a question for experiments to resolve.  Our analysis of algorithm \textsf{S} suggests that doing a non-constant number of shortcuts per round is likely to degrade performance.

\section{Related Work}\label{rw_section}

In this section we review previous work related to ours. We have presented our
results first, since they provide insights into the related work.
As far as we can tell, all our algorithms are novel and simpler than previous algorithms, although they are based on some of the previous algorithms.

Two different communities have worked on concurrent connected components algorithms, in two overlapping eras.
First, theoretical computer scientists developed provably efficient algorithms for various versions of the PRAM model.
This work began in the late 1970's and reached a natural conclusion in the work of Halperin and Zwick \cite{DBLP:journals/jcss/HalperinZ96, halperin2001optimal}, who gave $O(\lg n)$-step, $O(m)$-work randomized algorithms for the EREW (exclusive read, exclusive write) PRAM.
Their second algorithm finds spanning trees of the components.
The EREW PRAM is the weakest variant of the PRAM model, and finding connected components in this model requires $\Omega(\lg n)$ steps \cite{DBLP:journals/siamcomp/CookDR86}.
To solve the problem sequentially takes $O(m)$ time, so the Halperin-Zwick algorithms minimize both the number of steps and the total work (number of steps times the number of processes).
Whether there is a deterministic EREW PRAM algorithm with the same efficiency remains an open problem.

Halperin and Zwick's paper \cite{halperin2001optimal} contains a table listing results preceding theirs, and we refer the reader to their paper for these results.
Our interest is in simple algorithms for a more powerful computational model, so we content ourselves here with discussing simple labeling algorithms related to ours.  (The Halperin-Zwick algorithms and many of the preceding ones are \emph{not} simple.)
First we review variants of the PRAM model and how they relate to our algorithmic framework.

The three main variants of the PRAM model, in increasing order of strength, are EREW, CREW (concurrent read, exclusive write), and CRCW (concurrent read, concurrent write).
The CRCW PRAM has four standard versions that differ in how they handle write conflicts:
(\romannumeral1) COMMON: all writes to the same location at the same time must be of the same value;
(\romannumeral2) ARBITRARY: among concurrent writes to the same location, an arbitrary one succeeds;
(\romannumeral3) PRIORITY: among concurrent writes to the same location, the one done by the highest-priority process succeeds;
(\romannumeral4) COMBINING: values written concurrently to a given location are combined using some symmetric function.
As discussed in \S{\ref{af_section}},
our algorithms can be implemented on a COMBINING CRCW PRAM, with minimization as the combining function.

An early and important theoretical result is the $O(\lg^2 n)$-step CREW PRAM algorithm of Hirschberg, Chandra, and Sarwate \cite{DBLP:journals/cacm/HirschbergCS79}.
Algorithm \textsf{S} is a simplification of
their algorithm. They represent the graph by an adjacency matrix, but it is easy to translate their basic algorithm into our framework. Their algorithm alternates connect steps with repeated shortcuts.
To do connection they use a variant of \emph{parent-connect} that we call \emph{strong-parent-connect}.
It concurrently sets $x.p$ for each vertex $x$ equal to the minimum $w.p \neq x$ such that there is an edge $\{v, w\}$ with $v.p = x$;
if there is no such edge, $x.p$ does not change.
The following pseudocode implements this method in our framework: \\

\emph{strong-parent-connect}: \\
\indent\indent \textbf{for} each vertex $v$ \textbf{do} \\
\indent\indent\indent $v.n = \infty$ \\
\indent\indent \textbf{for} each edge $\{v, w\}$ \textbf{do} \\
\indent\indent\indent \textbf{if} $v.p \ne w.p$ \textbf{then} \\
\indent\indent\indent\indent $v.p.n = \min\{v.p.n, w.p\}$ \\
\indent\indent\indent\indent $w.p.n = \min\{w.p.n, v.p\}$ \\
\indent\indent\indent \textbf{for} each vertex $v$ \textbf{do} \\
\indent\indent\indent\indent \textbf{if} $v.n \ne \infty$ \textbf{then} \\
\indent\indent\indent\indent\indent $v.p = v.n$  \\

This version of connection can make a larger vertex the parent of a smaller one. Thus their algorithm does not do minimum labeling. Furthermore it can create parent cycles of length two, which Hirschberg et al. eliminate in a cleanup step at the end of each round. To do the cleanup it
suffices to concurrently set $v.p = v$ for each vertex such that $v.p > v$ and $v.p.p = v$.
Their algorithm is one of two we have found in the literature that can create parent cycles.

Although they do not say this, we think the reason Hirschberg et al. used \emph{strong-parent-connect} was to guarantee that each tree links with another tree in each round.
This gives them an $O(\lg n)$ bound on the number of rounds and an $O(\lg^2 n)$ bound on the number of steps, since there are $O(\lg n)$ shortcuts per round.
Our simpler algorithm \textsf{S} uses \emph{parent-connect} in place of \emph{strong-parent-connect}, making it a minimum labeling algorithm and eliminating the cleanup step.
Although \emph{parent-connect} does not guarantee that each tree links with another tree every round, it does guarantee such linking every two rounds, giving us the same $O(\lg^2 n)$ step bound as Hirschberg et al. See the proof of Theorem~\ref{thm_s_log}.


The first $O(\lg n)$-step PRAM algorithm was that of Shiloach and Vishkin \cite{DBLP:journals/jal/ShiloachV82}.
It runs on an ARBITRARY CRCW PRAM, as do the other algorithms we discuss, except as noted. The following is a version of their algorithm SV in our framework: \\

Algorithm SV: \\
\indent\indent\textbf{repeat} \\
\indent\indent\indent \{\emph{shortcut}; \emph{arb-parent-root-connect}; \emph{stagnant-parent-root-connect}; \emph{shortcut}\} \\
\indent\indent \textbf{until} no $v.p$ changes \\


\indent\indent \emph{arb-parent-root-connect}: \\
\indent\indent\indent \textbf{for} each vertex $v$ \textbf{do} \\
\indent\indent\indent\indent $v.o = v.p$ \\
\indent\indent\indent \textbf{for} each edge $\{v, w\}$ \textbf{do} \\
\indent\indent\indent\indent \textbf{if} $v.o > w.o$ and $v.o = v.o.o$ \textbf{then} \\
\indent\indent\indent\indent\indent $v.o.p = w.o$ \\
\indent\indent\indent\indent \textbf{else if} $w.o = w.o.o$ \textbf{then} \\
\indent\indent\indent\indent\indent $w.o.p = v.o$ \\


\indent\indent \emph{stagnant-parent-root-connect}: \\
\indent\indent\indent \textbf{for} each vertex $v$ \textbf{do} \\
\indent\indent\indent\indent $v.o = v.p$ \\
\indent\indent\indent \textbf{for} each edge $\{v, w\}$ \textbf{do} \\
\indent\indent\indent\indent \textbf{if} $v.o \ne w.o$ \textbf{then} \\
\indent\indent\indent\indent\indent \textbf{if} $v.o$ is a stagnant root \textbf{then} \\
\indent\indent\indent\indent\indent\indent $v.o.p = w.o$ \\
\indent\indent\indent\indent\indent \textbf{else if} $w.o$ is a stagnant root \textbf{then} \\
\indent\indent\indent\indent\indent\indent $w.o.p = v.o$ \\

Whereas \emph{parent-root-connect} updates the parent of each root $x$ to be the \emph{minimum} $y$ such that there is an edge $\{v, w\}$ with $x = v.o$ and $y = w.o < v.o$ if there is such an edge, \emph{arb-parent-root-connect} replaces the parent of each such root by an \emph{arbitrary} such $y$.
Arbitrary resolution of write conflicts suffices to implement the latter method, but not the former.

Shiloach and Vishkin define a root to be \emph{stagnant} if its tree is not changed by the first two steps of the main loop (the first shortcut and the \emph{arb-parent-root-connect}).
Their algorithm has additional steps to keep track of stagnant roots.
Method \emph{stagnant-parent-root-connect} updates the parent of each stagnant root $x$ to be an arbitrary $y$ such that there is an edge $\{v, w\}$ with $x = v.o$ and $y = w.o \neq v.o$ if there is such an edge.
The definition of ``stagnant'' implies that no two stagnant trees are connected by an edge.

Algorithm SV does not do minimum labeling, since \emph{stagnant-parent-root-connect} can make a larger vertex the parent of a smaller one.  Nevertheless, the algorithm creates no cycles, although the proof of this is not entirely straightforward.  The efficiency analysis is also not straightforward.

Algorithm \textsf{R} is algorithm SV with the third and fourth steps of the main loop deleted and the second step modified to resolve concurrent writes by minimum value instead of arbitrarily.
Shiloach and Vishkin state that one shortcut can be deleted from their algorithm without affecting its asymptotic efficiency.
They included the third step for two reasons:
(\romannumeral1) their analysis examines one round at a time, requiring that every tree change in every round, and
(\romannumeral2) if the third step is deleted, the algorithm can take $\Omega(n)$ steps on a graph that is a tree with edges $\{i, n\}$ for $i \in [n-1]$.
This example strongly suggests that to obtain a simpler algorithm one needs to use a more powerful model of computation, as we have done by using minimization to resolve write conflicts.

Awerbuch and Shiloach presented a slightly simpler $O(\lg n)$-step algorithm and gave a simpler efficiency analysis \cite{DBLP:journals/tc/AwerbuchS87}.
Our analysis of algorithms \textsf{R} and \textsf{RA} in \S{\ref{subsection_lsb}} uses a variant of their potential function.
Their algorithm is algorithm SV with the first shortcut deleted and the two connect steps modified to update only parents of roots of flat trees.  The computation needed to keep track of flat tree roots is simpler than that needed in algorithm SV to keep track of stagnant roots.

An even simpler but randomized $O(\lg n)$-step algorithm was proposed by Reif \cite{reif1984optimal}: \\

Algorithm Reif: \\
\indent\indent \textbf{repeat} \\
\indent\indent\indent \{for each vertex flip a coin; \emph{random-parent-connect}; \emph{shortcut}\} \\
\indent\indent \textbf{until} no $v.p$ changes \\


\indent\indent \emph{random-parent-connect}: \\
\indent\indent\indent \textbf{for} each vertex $v$ \textbf{do} \\
\indent\indent\indent\indent $v.o = v.p$ \\
\indent\indent\indent \textbf{for} each edge $\{v, w\}$ \textbf{do} \\
\indent\indent\indent\indent \textbf{if} $v.o$ flipped heads and $w.o$ flipped tails \textbf{then} \\
\indent\indent\indent\indent\indent $v.o.p = w.o$ \\
\indent\indent\indent\indent \textbf{else if} $w.o$ flipped heads and $v.o$ flipped tails \textbf{then} \\
\indent\indent\indent\indent\indent $w.o.p = v.o$ \\

Reif's algorithm keeps the trees flat, making the algorithm monotone, although it does not do minimum labeling.
Although it is randomized, Reif's algorithm is simpler than those of Shiloach and Vishkin, but \textsf{R} and \textsf{RA} are even simpler and are deterministic.

We know of one algorithm other than that of Hirschberg et al. \cite{DBLP:journals/cacm/HirschbergCS79} that does not maintain acyclicity.
This is the algorithm of Johnson and Metaxis \cite{DBLP:journals/jal/JohnsonM95}.
Their algorithm runs in $O((\lg n)^{3/2})$ steps on an EREW PRAM.
To eliminate any cycles it creates, it does a form of shortcutting.

Algorithms that run on a more restricted form of PRAM, or use fewer processes (and thereby do less work) use various kinds of edge alteration, edge addition, and edge deletion, along with techniques to resolve read and write conflicts. Such algorithms are much more complicated than those we have considered.
Again we refer the reader to \cite{DBLP:journals/jcss/HalperinZ96, halperin2001optimal} for results and references.


The second era of concurrent connected components algorithms was that of the experimentalists.  It began in the 1990's and continues to the present.  Experimentation has expanded greatly with the growing importance of huge graphs representing the internet, the world-wide web, friendship connections, and other symmetric relations, as well as the development of cloud computing frameworks.
These trends make concurrent algorithms for connected components both practical and useful.  The general approach of the experimentalists has been to take one or more existing algorithms, possibly simplify or modify them, implement the resulting suite of algorithms on one or more computing platforms, and report the results of experiments done on some collection of graphs.
Examples of such studies include
\cite{DBLP:conf/dimacs/GoddardKP94, DBLP:conf/spaa/Greiner94, hsu1997parallel, DBLP:journals/pvldb/YanCXLNB14, DBLP:conf/hotos/McSherryIM15, DBLP:conf/wsdm/StergiouRT18}.

Some of these papers make claims about the theoretical efficiency of algorithms they propose, but several of these claims are incorrect or unjustified.
We give some examples.
The first is a paper by Greiner \cite{DBLP:conf/spaa/Greiner94} in which he claims an $O(\lg^2 n)$ step bound for his ``hybrid'' algorithm.

Greiner's description of this algorithm is incomplete.
The algorithm is a modification of the algorithm of Hirschberg et al. \cite{DBLP:journals/cacm/HirschbergCS79}.
Each round does a form of direct connect followed by repeated shortcuts followed by an alteration. Since repeated shortcuts guarantee that all trees are flat at the beginning of each round, this is equivalent to using a version of \emph{parent-connect} and not doing alteration.
The main novelty in his algorithm is that alternate rounds use
maximization instead of minimization in the connect step. He does not specify exactly how the connect step works. There are at least two possibilities. One is to use \emph{direct-connect}, but in alternate rounds replace min by max. The resulting algorithm is a min-max version of algorithm \textsf{S}. The second is to use the following strong version of \emph{direct-connect}, but in alternate rounds replace min by max and $\infty$ by $-\infty$: \\

\emph{strong-direct-connect}: \\
\indent\indent \textbf{for} each vertex $v$ \textbf{do} \\
\indent\indent\indent $v.n = \infty$ \\
\indent\indent \textbf{for} each edge $\{v, w\}$ \textbf{do} \\
\indent\indent\indent $v.n = \min\{v.n, w\}$ \\
\indent\indent\indent $w.n = \min\{w.n, v\}$ \\
\indent\indent \textbf{for} each vertex $v$ \textbf{do} \\
\indent\indent\indent \textbf{if} $v.n \ne \infty$ \textbf{then} \\
\indent\indent\indent\indent $v.p = v$ \\

The resulting algorithm is a min-max version of the Hirschberg et al. algorithm. 
Greiner claims an $O(\lg n)$ bound on the number of rounds and an $O(\lg^2 n)$ bound on the number of steps. But these bounds do not hold for the algorithm that uses the min-max version of \emph{direct-connect}:
on the bad example of Shiloach and Vishkin consisting of an unrooted tree with vertex $n$ adjacent to vertices $1$ through $n-1$, the algorithm takes $\Omega(n)$ steps.
This example has a high-degree vertex, but there is a simple example whose vertices are of degree at most three, consisting of a path of odd vertices $1, 3, 5, \dots, n$ with each even vertex $i$ adjacent to $i + 1$.
On the other hand, the algorithm that uses the min-max version of \emph{strong-direct-connect} can create parent cycles of length two, which must be eliminated by a cleanup as in the Hirschberg et al. algorithm.
Greiner says nothing about eliminating cycles. We conclude that either his step bound is incorrect or his algorithm is incorrect.

At least one other work reproduces Greiner's error: Soman et al. \cite{DBLP:conf/ipps/SomanKN10, DBLP:journals/ppl/SomanKN10} propose a modification of Greiner's algorithm intended for implementation on a GPU model. Their algorithm is the inefficient version of Greiner's algorithm, modified to use \emph{parent-connect} instead of \emph{direct-connect} and without alteration. Their specific implementation of the connect step is as follows: \\

\emph{alternate-connect}: \\
\indent\indent \textbf{for} each edge $\{v, w\}$ \textbf{do} \\
\indent\indent\indent \textbf{if} $v.p \ne w.p$ \textbf{then} \\
\indent\indent\indent\indent $x = \min\{v.p, w.p\}$ \\
\indent\indent\indent\indent $y = \max\{v.p, w.p\}$ \\
\indent\indent\indent\indent \textbf{if} round is even \textbf{then} \\
\indent\indent\indent\indent\indent $y.p = x$ \\
\indent\indent\indent\indent \textbf{else} $x.p = y$ \\

Soman et al. say nothing about how to resolve concurrent writes. If this resolution is arbitrary, or by minimum in the even rounds and by maximum in the odd rounds, then the algorithm takes $\Omega(n)$ steps on the examples mentioned above.

Algorithm \textsf{S}, and the equivalent algorithm that uses \emph{direct-connect} and alteration, are simpler than the algorithms of Greiner and Soman et al. and have guaranteed $O(\lg^2 n)$ step bounds.
We conclude that alternating minimization and maximization adds complication without improving efficiency, at least in theory.

Another paper that has an invalid efficiency bound as a result of not handling concurrent writes carefully is that of Yan et al. \cite{DBLP:journals/pvldb/YanCXLNB14}.
They consider algorithms in the PREGEL framework \cite{DBLP:conf/sigmod/MalewiczABDHLC10}, which is a graph-processing platform designed on top of the MPC model.
All the algorithms they consider can be expressed in our framework.
They give an algorithm obtained from algorithm SV by deleting the first shortcut and replacing the second connect step by the first connect step of Awerbuch and Shiloach's algorithm.  In fact, the second connect step does nothing, since any parent update it would do has already been done by the first connect step.
That is, this algorithm is equivalent to algorithm SV with the first shortcut and the second connect step deleted.  Their termination condition, that all trees are flat, is incorrect, since there could be two or more flat trees in the same component.
They claim an $O(\lg n)$ bound on steps, but since they assume arbitrary resolution of write conflicts, the actual step bound is $\Theta(n)$ by the example of Shiloach and Vishkin.

A third paper with an analysis gap is that of Stergio, Rughwani, and Tsioutsiouliklis \cite{DBLP:conf/wsdm/StergiouRT18}.
They present an algorithm that we call SRT, whose main loop expressed in our framework is the following: \\

Algorithm SRT: \\
\indent\indent \textbf{repeat} \\
\indent\indent\indent \textbf{for} each vertex $v$ \textbf{do} \\
\indent\indent\indent\indent $v.o = v.p$ \\
\indent\indent\indent\indent $v.n = v.p$ \\
\indent\indent\indent \textbf{for} each edge $\{v, w\}$ \textbf{do} \\
\indent\indent\indent\indent \textbf{if} $v.o > w.o$ \textbf{then} \\
\indent\indent\indent\indent\indent $v.n = \min\{v.n, w.o\}$ \\
\indent\indent\indent\indent \textbf{else} $w.n = \min\{w.n, v.o\}$ \\
\indent\indent\indent \textbf{for} each vertex $v$ \textbf{do} \\
\indent\indent\indent\indent $v.o.p = \min\{v.o.p, v.n\}$ \\
\indent\indent\indent \textbf{for} each vertex $v$ \textbf{do} \\
\indent\indent\indent\indent $v.p = \min\{v.p, v.n.o\}$ \\
\indent\indent \textbf{until} no $v.p$ changes \\

This algorithm does an extended form of connection combined with a variant of shortcutting that combines old and new parents.
It is not monotone.
Stergio et al. implemented this algorithm on the Hronos computing platform and successfully solved problems with trillions of edges.
They claimed an $O(\lg n)$ step bound for the algorithm, but we are unable to make sense of their analysis.
Their paper motivated our work.  We are so far unable to prove any interesting step bound for algorithm SRT, but using the techniques of \S{\ref{subsec_d_bound}} we can prove an $O(\lg^2 n)$ step bound for the following variant: \\

Algorithm \textsf{H}: \\
\indent\indent \textbf{repeat} \\
\indent\indent\indent \textbf{for} each vertex $v$ \textbf{do} \\
\indent\indent\indent\indent $v.o = v.p$ \\
\indent\indent\indent\indent $v.n = v.p$ \\
\indent\indent\indent \textbf{for} each edge $\{v, w\}$ \textbf{do} \\
\indent\indent\indent\indent \textbf{if} $v.o > w.o$ \textbf{then} \\
\indent\indent\indent\indent\indent $v.o.n = \min\{v.o.n, w.o\}$ \\
\indent\indent\indent\indent\indent \textbf{else} $w.o.n = \min\{w.o.n, v.o\}$; \\
\indent\indent\indent\indent \textbf{for} each vertex $v$ \textbf{do} \\
\indent\indent\indent\indent\indent $v.p = \min\{v.p, v.n, v.o.n, v.n.n\}$ \\
\indent\indent \textbf{until} no $v.p$ changes \\

Algorithm \textsf{H} is algorithm \textsf{P} with the shortcut step replaced by a hybrid shortcutting step in which both the current parent of the current parent and the current parent of the old parent are candidates to be the new parent.  We omit the analysis of algorithm \textsf{H}.

A final paper with an interesting algorithm but incorrect analysis is that of Burkhardt \cite{DBLP:journals/corr/absPB}.
The main novelty in Burkhardt's algorithm is to replace each edge $\{v, w\}$ by a pair of oppositely directed arcs $(v, w)$ and $(w, v)$ and to use \emph{asymmetric} alteration:
he replaces $(v, w)$ by $(w, v.p)$ instead of $(v.p, w.p)$ (unless $w = v.p$).
This idea allows him to combine connecting and shortcutting in a natural way.  (Burkhardt claims that his algorithm does not do shortcutting, but it does, implicitly.)
Burkhardt does not give an explicit stopping rule, saying only, ``This is repeated until all labels converge.''
An iteration can alter arcs without changing any parents, so one must specify the stopping rule carefully.
The following is a version of the main loop of Burkhardt's algorithm with the parent updates and the arc alterations disentangled, and which stops when there is one root and all other vertices are leaves: \\

Algorithm \textsf{B}: \\
\indent\indent \textbf{repeat} \\
\indent\indent\indent \textbf{for} each arc $(v, w)$ \textbf{do} \\
\indent\indent\indent\indent \textbf{if} $v > w$ \textbf{then} \\
\indent\indent\indent\indent\indent $v.p = \min\{v.p, w\}$ \\
\indent\indent\indent \textbf{for} each arc $(v, w)$ \textbf{do} \\
\indent\indent\indent\indent \textbf{if} $v.p \neq w$ \textbf{then} \\
\indent\indent\indent\indent\indent replace $(v, w)$ by $(w, v.p)$ \\
\indent\indent\indent\indent \textbf{else} delete $(v, w)$ \\
\indent\indent\indent \textbf{for} each vertex $v$ \textbf{do} \\
\indent\indent\indent\indent \textbf{if} $v.p \neq v$ \textbf{then} \\
\indent\indent\indent\indent\indent add arc $(v.p, v)$\\
\indent\indent \textbf{until} every arc $(v, w)$ has $v.p = w.p$ and $v.p \in \{v, w\}$ \\

Burkhardt claims that his algorithm takes $O(\lg d)$ steps, which would be remarkable if true.
Unfortunately, a long skinny grid is a counterexample, as shown in \cite{DBLP:conf/focs/Andoni}.
Burkhardt also claimed that the number of arcs existing at any given time is at most $2m + n$.
The version above has a $2m$ upper bound on the number of arcs.
Two small changes in the algorithm reduce the upper bound on the number of arcs to $m$ and make the shortcutting more efficient:
replace each original edge $\{v, w\}$ by \emph{one} arc $(\max\{v, w\}, \min\{v, w\})$, and in the loop over the vertices replace ``add arc $(v.p, v)$'' by ``add arc $(v, v.p.p)$.''
We call the resulting algorithm \textsf{AA}, for \emph{asymmetric alteration}. The following pseudocode implements this algorithm: \\

Algorithm \textsf{AA}: \\
    \indent\indent \textbf{for} each vertex $v$ \textbf{do} $v.p = v$ \\
    \indent\indent \textbf{for} each edge $\{v, w\}$ \textbf{do} \\
    \indent\indent\indent replace $\{v,w\}$ by arc $(\max\{v, w\}, \min\{v, w\})$ \\
    \indent\indent \textbf{repeat} \\
    \indent\indent\indent \textbf{for} each arc $(v, w)$ \textbf{do} \\
    \indent\indent\indent\indent $v.p = \min\{v.p, w\}$ \\
    \indent\indent\indent \textbf{for} each arc $(v, w)$ \textbf{do} \\
    \indent\indent\indent\indent delete $(v, w)$ \\
    \indent\indent\indent\indent \textbf{if} $w \ne v.p$ \textbf{then} \\
    \indent\indent\indent\indent\indent add arc $(w, v.p)$ \\
    \indent\indent\indent \textbf{for} each vertex $v$ \textbf{do} \\
    \indent\indent\indent\indent \textbf{if} $v \ne v.p$ \textbf{then} \\
    \indent\indent\indent\indent\indent add arc $(v, w.p)$ \\
    \indent\indent \textbf{until} no arc $(v, w)$ has $w \ne v.p$ \\

A version of Algorithm \textsf{AA} was proposed to us by Yu-Pei Duo [private communication, 2018].
The techniques of \S{\ref{subsec_d_bound}} extend to give an $O(\lg^2 n)$ step bound for \textsf{AA}, \textsf{B}, and Burkhardt's original algorithm.
We omit the details.

Very recently, theoreticians have become interested in concurrent algorithms for connected components again, with the aim of obtaining a step bound logarithmic in $d$ rather than $n$, for a suitably powerful model of computation.
The first, breakthrough result in this direction was that of Andoni et al. \cite{DBLP:conf/focs/Andoni}.
They gave a randomized algorithm that takes $O(\lg d \lg \log_{m/n} n)$ steps in the MPC model.
Their algorithm uses graph densification based on the distance-doubling technique of \cite{DBLP:conf/icde/RastogiMCS13}, controlled to keep the number of edges linearly bounded.
Behnezhad et al. \cite{DBLP:journals/corr/abs-1910-05385} improved the result of Andoni et al. by reducing the number of steps to $O(\lg d + \lg\log_{m/n} n)$.  Their algorithm can be implemented in the MPC model or on a very powerful version of the CRCW PRAM that supports a ``multiprefix'' operation.
In a recent work \cite{DBLP:journals/corr/LTZ20} we show that this algorithm and that of Andoni et al. can be simplified and implemented on an ARBITRARY CRCW PRAM.

\section{Remarks}\label{remark_section}

We have presented several very simple label-update algorithms to compute connected components concurrently.
Our best bounds, of $O(\lg n)$ steps and $O(m \lg n)$ work, are for two related monotone algorithms, \textsf{R} and \textsf{RA}.
For two other algorithms, \textsf{A}, which is non-monotone, and \textsf{S}, which keeps all trees flat by doing repeated shortcuts, our bounds are $O(\lg^2 n)$ steps and $O(m \lg^2 n)$ work, which are tight for \textsf{S} but maybe not for \textsf{A}.
We have also pointed out errors in previous analyses of similar algorithms.

Our analysis of these algorithms is novel in that it extends over several rounds of the main loop, unlike previous analyses that consider only one round at a time.
Our analysis of \textsf{A} combines new ideas with an idea from the analysis of disjoint set union algorithms.
As mentioned in \S{\ref{rw_section}}, this analysis extends to give analyses of two other algorithms in the literature, and variants of these algorithms.

Our results illustrate the subtleties of even simple algorithms.
A number of theoretical questions remain open, notably, determining tight asymptotic step bounds for algorithms \textsf{P}, \textsf{A}, \textsf{H}, SRT, \textsf{AA}, and \textsf{B}.
For \textsf{P} and SRT we know nothing interesting:
our techniques seem too weak to derive a
poly-logarithmic step bound for an algorithm such as \textsf{P} or SRT that is non-monotone and in which trees can be passive for an indefinite number of rounds.
For \textsf{A}, \textsf{H}, \textsf{AA}, and \textsf{B} we have a bound of $O(\lg^2 n)$ steps but the lower bound is $\Omega(\lg n)$.

All our algorithms are simple enough to merit experimental study.
We leave this for future work.
Interesting questions to address are whether monotonicity helps in practice and how many shortcuts should be done per round.

There is a natural way to convert each of the deterministic algorithms we have considered into a randomized algorithm: number the vertices from $1$ to $n$ uniformly at random and identify the vertices by number.
In an application, one may get such randomization for free, for example if a hash table stores the vertex identifiers. It is natural to study the efficiency that results from such randomization.
We are doing this, in ongoing work with Eitan Zlatin. We can prove a high-probability $O(\lg n)$ or $O(\lg^2 n)$ step bound for the randomized version of most of the algorithms we have presented, and in particular an $O(\lg n)$ bound for algorithm \textsf{A}, improving our $O(\lg^2 n)$ worst-case bound. We shall report on these results in the future.

We have assumed global synchronization.
The problem becomes much more challenging in an asynchronous setting.
One of the authors and a colleague have studied asynchronous concurrent disjoint set union \cite{DBLP:conf/podc/JayantiT16}, the incremental version of the connected components problem.
Their algorithms can be used to find connected components asynchronously.

An interesting extension of the connected components problem is to construct a spanning tree of each component.  It is easy to extend algorithms \textsf{R}, \textsf{RA}, and \textsf{S} to do this:
when an edge causes a root to become a child, add the corresponding original edge to the spanning forest.
Extending non-monotone algorithms such as \textsf{A} to construct spanning trees seems a much bigger challenge.

\paragraph{Acknowledgements and Correction.}
We thank Dipen Rughwani, Kostas Tsioutsiouliklis, and Yunhong Zhou for telling us about \cite{DBLP:conf/wsdm/StergiouRT18}, for extensive discussions about the problem and our algorithms, and for insightful comments on our early results.

In the preliminary version of this work \cite{liu_tarjan}, we claimed an $O(\lg^2 n)$ step bound for algorithm \textsf{P} and for a related algorithm \textsf{E} (which we have omitted from the current paper).
We thank Pei-Duo Yu for discovering that Lemma 6 in \cite{liu_tarjan} does not hold for \textsf{P} and \textsf{E}, invalidating our proof of Theorem 13 in \cite{liu_tarjan} for these algorithms. 


\bibliographystyle{alpha}
\bibliography{CACC}

\end{document}